\documentclass[reqno]{amsart}
\usepackage{geometry}                
\usepackage[parfill]{parskip}    
\usepackage{graphicx}
\usepackage{amssymb}
\usepackage{epstopdf}
\usepackage{stmaryrd}
\usepackage{esint}
\usepackage{mathrsfs} 
\usepackage[english]{babel}
\usepackage{amssymb,
enumerate}
\usepackage{color}
\usepackage[latin1]{inputenc}
\usepackage[T1]{fontenc}
\usepackage{amsfonts}
\usepackage{amssymb}
\usepackage{amsmath}
\usepackage{amsthm}
\usepackage{graphicx,xcolor}
\usepackage{afterpage}
\usepackage[colorlinks=true, linkcolor=red, citecolor=cyan, urlcolor=green]{hyperref} 
\usepackage{bbm}
\usepackage{latexsym}

\def\be{\begin{equation}}
\def\ee{\end{equation}}
\def\bea{\begin{eqnarray}}
\def\eea{\end{eqnarray}}
\def\bel{\begin{align}}
\def\el{\end{align}}

\def\nn{\nonumber}

\def\R{\mathbb{R}}
\def\s{\sigma}
\def\a{\alpha}
\def\e{\varepsilon}

\def\G{\Gamma}
\def\b{\beta}
\def\D{\Delta}

\newcommand{\ie}{\textit{i.e.\ }}

\newcommand{\nocontentsline}[3]{}
\newcommand{\tocless}[2]{\bgroup\let\addcontentsline=\nocontentsline#1{#2}\egroup}

\DeclareMathSymbol{\leqslant}{\mathalpha}{AMSa}{"36} 
\DeclareMathSymbol{\geqslant}{\mathalpha}{AMSa}{"3E} 
\DeclareMathSymbol{\eset}{\mathalpha}{AMSb}{"3F}     
\renewcommand{\leq}{\;\leqslant\;}                   
\renewcommand{\geq}{\;\geqslant\;}                   

\def\Err{\operatorname{Error}}

\DeclareMathOperator{\RS}{RS}

\def\ie{\textit{i.e. }}

\def\a{\alpha}
\def\e{\varepsilon}

\def\G{\Gamma}
\def\b{\beta}
\def\D{\Delta}

\def\s{\sigma}

\def\R{\mathbb{R}}

\theoremstyle{plain}
\newtheorem{proposition}{Proposition}
\newtheorem{lemma}{Lemma}

\definecolor{light}{gray}{.9}

\title
{A remark on the spherical bipartite spin glass}

\author{Giuseppe Genovese}
\address{Institute of Mathematics, University of Zurich, Winterthurerstrasse 190, 8057 Zurich, Switzerland.}
\email{giuseppe.genovese@math.uzh.ch}

\date{\today}                     

\begin{document}
\maketitle

\begin{abstract}
In \cite{AC} Auffinger and Chen proved a variational formula for the free energy of the spherical bipartite spin glass in terms of a global minimum over the overlaps. We show that a different optimisation procedure leads to a saddle point, similar to the one achieved for models on the vertices of the hypercube. 
\vspace{0.5cm}

\end{abstract}

\section{Introduction}

Let $\s_{N}(dx)$ denote the uniform probability measure on $S^N:=\{x\in\R^N\,:\,\|x\|^2_2=N\}$, where $\|x\|_2$ is the Euclidean norm. 
For $x:=(x_1,\ldots x_{N_1})\in\R^{N_{1}}$ and $y:=(y_1,\ldots,y_{N_2})\in\R^{N_{2}}$ the bipartite spin glass is defined by the energy function
\be\label{eq:RBM}
H_{N_1,N_2}(x,y;\xi):=-\frac{1}{\sqrt{N}}\sum_{j=1}^{N_2}\sum_{i=1}^{N_1} \xi_{ij}x_iy_j\,. 
\ee
Here $\{\xi_{ij}\}_{i\in[N_1],j\in[N_{2}]}$ are $\mathcal{N}(0,1)$ i.i.d.\ quenched r.vs.\ and we set $N:=N_1+N_2$.
The object of interest of this note is the free energy
\be\label{eq:A-sfer}
A_{N_1,N_2}(\b,\xi):=\frac1N\log \int \s_{N_1}(dx)\s_{N_2}(dy)\exp(-\b H_{N_1,N_2}(x,y;\xi)-b_1(x,1)-b_2(y,1))\,
\ee
in the limit in which $N_1,N_2\to\infty$ with $N_1/N\to\a\in(0,1)$. Here $\b\geq0$ is the inverse temperature, $b_1,b_2\in\R$ are external fields and $(\cdot,\cdot)$ denotes the Euclidean inner product. 
By concentration of Lipschitz functions of Gaussian random variables one reduces to study the average free energy $A_{N_1,N_2}(\b):=E[A_{N_1,N_2}(\b,\xi)]$, whose limit we denote by $A(\a,\b)$. 

Auffinger and Chen proved in \cite{AC} the following variational formula for $A(\a,\b)$ for $\b$ small enough
\bea
A(\a,\b)&=&\min_{q_1,q_2\in[0,1]^2} P(q_1,q_2)\label{eq:ACh}\\
P(q_1,q_2)&=&\frac{\b^2\a(1-\a)}{2}(1-q_1q_2)+\frac\a2\left(b_1^2(1-q_1)+\frac{q_1}{1-q_1}+\log(1-q_1)\right)\nn\\
&+&\frac{1-\a}2\left(b_2^2(1-q_2)+\frac{q_2}{1-q_2}+\log(1-q_2)\right)\,\label{eq:P-ACh}
\eea
(the normalisation in (\ref{eq:RBM}) leads to different constants w.r.t. \cite{AC}). 
The above formula was successively proved to hold in the whole range of $\beta\geq0$ in \cite{baik, sferico}. Yet these proofs are indirect, as in both cases one obtains a formula for the free energy and then verifies a posteriori (analytically for \cite{baik} and numerically \cite{DT} for \cite{sferico}) that it coincides with (\ref{eq:ACh}). We just mention that the results in \cite{AC} have been recently extended in \cite{mc, tucaphd} for the complexity and in \cite{bates1,bates2} for the free energy.

The convex variational principle found by Auffinger and Chen appears to be in contrast with the $\min\max$ characterisation given in \cite{bip,GG} for models on the vertices of the hypercube (see also \cite{NN} for the Hopfield model). The aim of this note is to show that the Auffinger and Chen formula can be equivalently expressed in terms of a $\min\max$.


One disadvantage of the spherical prior is that the associated moment generating function 
\be\label{eq_Gamma}
\Gamma_N(h):=\frac1N\log \int\s_N(dx)e^{(h,x)}\,,\quad h\in\R^N\,, 
\ee
is not easy to compute. If $h$ is random with i.i.d.\ $\mathcal N(b,q)$ components it is convenient to set
\be\label{eq_Gamma-Lim}
\Gamma(b,q):=\lim_NE \Gamma_N(h)\,.
\ee
The so-called Crisanti-Sommers variational characterisation of it as $N\to\infty$ reads as follows.
\begin{lemma}\label{lemma_Gamma-as}
Let $b\in\R$, $q>0$, $h\in\R^N$ with i.i.d $\mathcal N(b,\sqrt q)$ components. Then
\be\label{eq:Gamma-as}
\Gamma(b,q)= \frac12\min_{r\in[0,1)}\left((b^2+q)(1-r)+\frac {r}{1-r}+\log(1-r)\right)\,
\ee
\end{lemma}
At the end of this note we give a simple proof of this statement, based on the method of \cite{sferico0, sferico}. We first get a variational characterisation of the moment generating function of a Gaussian distribution (whose variance is Legendre conjugate to $q$) and then use concentration of measure.   

A direct computation shows that the minimum of (\ref{eq:Gamma-as}) is attained for
\be\label{eq:tienamente}
\frac{r}{(1-r)^2}=q+b^2\,. 
\ee

A standard replica symmetric interpolation gives that for any $q_1,q_2\in[0,1]$
\bea
A _{N_1,N_2}(\b)\!\!\!&=&\!\!\!\frac{\b^2\a(1-\a)}{2}(1-q_1)(1-q_2)+(1-\a)\Gamma(b_2,\b^2\a q_1)
+\a \Gamma(b_1,\b^2(1-\a) q_2)\label{eq:sumrule-sfer}\\
&+&\Err_N(q_1,q_2)\,.\nn
\eea
The last summand is an error term whose specific form is not important here. What matters is that by \cite[Lemma 1]{AC} there is a choice of $(q_1,q_2)$ (see below) for which this remainder goes to zero as $N\to\infty$ if $\b$ is small enough. 
Combining (\ref{eq:Gamma-as}) and (\ref{eq:tienamente}) we can rewrite the first line of \eqref{eq:sumrule-sfer} as
\bea
\RS (q_1,q_2)&:=&\frac{\b^2\a(1-\a)}{2}(1-q_1)(1-q_2)\nn\\
&+&\frac{\b^2\a(1-\a)}{2}\left(\left(q_2+\frac{b_1^2}{\b^2(1-\a)}\right)(1-r_1)+\left(q_1+\frac{b_2^2}{\b^2\a}\right)(1-r_2)\right)\nn\\
&+&\frac\a2\frac{r_1}{1-r_1}+\frac\a2\log(1-r_1)+\frac{1-\a}2\frac{r_2}{1-r_2}+\frac{1-\a}2\log(1-r_2)\label{eq:RSsferica}\,,
\eea
under the condition
\be\label{eq:r1r2}
\frac{r_1}{(1-r_1)^2}=\b^2(1-\a) q_2+b_1^2\,,\quad \frac{r_2}{(1-r_2)^2}=\b^2\a q_1+b_2^2\,.
\ee
Here we used that there is a sequence $o_N\to0$ uniformly in $q_1,q_2,\b,\a$ such that
\be\label{eq:above}
\frac{\b^2\a(1-\a)}{2}(1-q_1)(1-q_2)+(1-\a)\Gamma_N(\b^2\a q_1)+\a \Gamma_N(\b^2(1-\a) q_2)=\RS (q_1,q_2)+o_N\,.
\ee
Indeed (\ref{eq:above}) follows easily once we use Lemma \ref{lemma_Gamma-as} for the limit of the functions $\G_N$ and we note that (\ref{eq:r1r2}) are the critical point equations related to the minimisation of (\ref{eq:Gamma-as}).

The main observation of this note is that (\ref{eq:RSsferica}) under (\ref{eq:r1r2}) is optimised as a $\min\max$.

\begin{proposition}\label{prop:minsferico}
Assume $b_1^2+b_2^2>0$. The function $\RS (q_1,q_2)$ has a unique stationary point $(\bar q_1, \bar q_2)$. It solves
\be\label{eq:sfer-crit}
\frac{q_2}{(1-q_2)^2}=\b^2\a q_1+b_1^2\,,\qquad \frac{q_1}{(1-q_1)^2}=\b^2(1-\a) q_2+b_2^2\,. 
\ee
Moreover
\be\label{eq:minmax}
\RS (\bar q_1,\bar q_2)=\min_{q_2\in[0,1]}\max_{q_1\in[0,1]} \RS (q_1,q_2)\,.
\ee
If $b_1=b_2=0$ and 
\be\label{eq:sfercritica_h=0-II}
\b^4\a(1-\a)<1\,
\ee
the origin is the unique solution of (\ref{eq:sfer-crit}) and
\be\label{eq:origin1}
\RS (0,0)=\min_{q_2\in[0,1]}\max_{q_1\in[0,1]} \RS (q_1,q_2)\,.
\ee
If $b_1=b_2=0$ and
\be\label{eq:sfercritica_h=0-III}
\b^4\a(1-\a)>1\,
\ee
there is a unique $(\bar q_1, \bar q_2)\neq (0,0)$ which solves (\ref{eq:sfer-crit}) and such that (\ref{eq:minmax}) holds. Moreover
\be\label{eq:origin2}
\RS (0,0)=\max_{q_2\in[0,1]}\max_{q_1\in[0,1]} \RS (q_1,q_2)\,.
\ee
\end{proposition}

The crucial point of \cite[Lemma 1]{AC} (for us) is that from the Latala argument \cite[Section 1.4]{Tal} it follows that the overlaps self-average as $N\to\infty$ at a point $(\tilde q_1, \tilde q_2)$ uniquely given by
\be\label{eq:r1r2AC}
\frac{\tilde q_1}{(1-\tilde q_1)^2}=\b^2(1-\a) \tilde q_2+b_1^2\,,\quad \frac{\tilde q_2}{(1-\tilde q_2)^2}=\b^2\a \tilde q_1+b_2^2\,,
\ee
which (see \cite[Lemma 7]{pan-sf}) are indeed asymptotically equivalent to 
\bea
q_{1,N}&:=&\frac1NE\left[\frac{\int\s_{N_1}(dy)\s_{N_1}(dy') (y,y') e^{\b\sqrt{q_2}(y+y',h)}}{\left(\int \s_{N_1}(dy) e^{\b\sqrt{q_2}(y,h)}\right)^2}\right]\,,\label{eq:t0}\\
\qquad q_{2,N}&:=&\frac1N E\left[\frac{\int\s_{N_2}(dx)\s_{N_2}(dx') (x,x') e^{\b\sqrt{q_1}(x+x',h)}}{\left(\int\s_{N_2}(dx) e^{\b\sqrt{q_1}(x,h)}\right)^2}\right]\,,\label{eq:t02}
\eea
naturally arising from the replica symmetric interpolation (here $h$ is random with i.i.d.\ $\mathcal N(0,1)$ entries). 
Comparing (\ref{eq:r1r2}) and (\ref{eq:r1r2AC}) readily implies that we can plug $(r_1,r_2)=(q_1,q_2)$ into (\ref{eq:RSsferica}) and obtain the convex function $P(q_1,q_2)$ of \cite[Theorem 1]{AC}, optimised by (\ref{eq:r1r2AC}). 

On the other hand, without using the Latala method one might still optimise (\ref{eq:RSsferica}) as a function of four variables, ignoring (\ref{eq:r1r2}). Taking derivatives first in $q_1,q_2$, the critical point equations (\ref{eq:sferDRS1}), \eqref{eq:sferDRS2} below select exactly $(q_1,q_2)=(r_1,r_2)$. This procedure is however unjustified a priori and this particular application of Latala's method legitimises the exchange in the order of the optimisation of the $q$ and the $r$ variables for small $\b$, which a posteriori can be extended to all $\b$ \cite{baik,sferico}.  

We stress that by itself the Latala method is not variational, it only gives the self-consistent equations for the critical points. It is the Crisanti-Sommers formula (\ref{eq:Gamma-as}) which makes it implicitly variational. Such a variational representation is not necessary in other cases of interest, for instance for the bipartite SK model (namely Hamiltonian (\ref{eq:RBM}) with $\pm1$ spins), for which one simply has the $\log\cosh$. Indeed in this case a direct use of the Latala method yields the validity of the $\min\max$ formula of \cite{bip} for $\b$ and $|b_1|,|b_2|$ small enough. The proof is essentially an exercise after \cite[Proposition 1.4.8]{Tal} and \cite[Formula (9)]{AC} and will not be reproduced here in details. The replica symmetric sum-rule for the free energy (analogue of formula (\ref{eq:sumrule-sfer})) reads as
\bea
A _{N_1,N_2}(\b)\!\!\!&=&\!\!\!\frac{\b^2\a(1-\a)}{2}(1-q_1)(1-q_2)+(1-\a)E\log\cosh(b_2+\b\sqrt{\a q_1}g)\nn\\
&+&\a E\log\cosh(b_1+\b\sqrt{(1-\a) q_2}g)\label{eq:sumrule-SK}\\
&+&\Err_N(q_1,q_2)\,,\nn
\eea
(here $g\sim\mathcal N(0,1)$) and the error term can be shown by the Latala method to vanish for small $\b,|b_1|, |b_2|$, if $(q_1,q_2)=(\bar q_1, \bar q_2)$ are given by
\be\label{eq:RS-SK}
\bar q_1=E[\tanh(b_1+\b\sqrt{(1-\a)\bar q_2} g)]\,,\qquad \bar q_2=E[\tanh(b_2+\b\sqrt{\a\bar q_1} g)]\,.
\ee
Therefore the free energy equals the first two lines on the r.h.s.\ of \eqref{eq:sumrule-SK} evaluated in $(q_1,q_2)=(\bar q_1, \bar q_2)$, which is the value attained at the $\min\max$, as shown in \cite{bip,GG}. 


\section{Proofs}

\begin{proof}[Proof of Proposition \ref{prop:minsferico}]
Assume first $b_1^2+b^2_2>0$. We differentiate (\ref{eq:RSsferica}) and by (\ref{eq:r1r2}) we get
\bea
\partial_{q_1} \RS &=&\frac{\b^2\a(1-\a)}{2}(q_2-r_2(q_1))\,\label{eq:sferDRS1}\\
\partial_{q_2} \RS &=&\frac{\b^2\a(1-\a)}{2}(q_1-r_1(q_2))\,.\label{eq:sferDRS2} 
\eea

The functions $r_1,r_2$ write explicitly as
\bea
r_1(q_2)&=&\frac{\sqrt{1+4(\b^2(1-\a)q_2+b_1^2)}-1}{\sqrt{1+4(\b^2(1-\a)q_2+b_1^2)}+1}\label{eq:r1explicito}\\
r_2(q_1)&=&\frac{\sqrt{1+4(\b^2\a q_1+b_2^2)}-1}{\sqrt{1+4(\b^2\a q_1+b_2^2)}+1}\label{eq:r2explicito}\,.
\eea
We easily see that $r_1,r_2$ are increasing from $r_1(0),r_2(0)>0$ (obviously computable by the formulas above)
to 1 and concave. Moreover we record for later use that if $b_1=b_2=0$ we have
\be\label{eq:derivate-r}
\frac{d}{dq_2}r_1(q_2)\Big|_{q_2=0}=\b^2(1-\a)\,,\quad \frac{d}{dq_1}r_2(q_1)\Big|_{q_1=0}=\b^2\a\,.
\ee

Now we take the derivative w.r.t. $q_1$ and note that the r.h.s.\ of (\ref{eq:sferDRS1}) is decreasing as a function of $q_1$, thus
$\partial^2_{q_1} \RS <0$. Therefore by the implicit function theorem there is a unique function $q_1$ such that $q_2=r_2(q_1)$. As a function of $q_2$, $q_1$ is non-negative, increasing and convex and it is $q_1(r_2(0))=0$. We set
\be\label{eq:RS1-gauss}
\RS _1(q_2):=\max_{q_1}\RS (q_1,q_2)=\RS (q_1(q_2),q_2)\,
\ee
and compute
\be\label{eq:derRS1gauss}
\partial_{q_2}\RS _1(q_2)=\frac{\b^2\a(1-\a)}{2}\left(q_1(q_2)-r_1(q_2)\right)\,. 
\ee
By the properties of the functions $q_1$ and $r_1$ it is clear that there is a unique intersection point $\bar q_2$; moreover $q_1\leq r_1$ for $q_2\leq \bar q_2$ and otherwise $q_1\geq r_1$. Therefore $\partial_{q_2}\RS _1(q_2)$ is increasing in a neighbourhood of $\bar q_2$ which allows us to conclude $\partial^2_{q_2}\RS _1>0$. This finishes the proof if $b_1^2+b_2^2>0$. 

If $b_1=b_2=0$ the origin is always a stationary point. It is unique if
\be\label{eq:sfercritica_h=0}
\left[\frac{d}{dq_1}r_2(q_1)\big|_{q_1=0}\right]^{-1}=\frac{d}{dq_2}q_1(q_2)\big|_{q_2=0}>\frac{d}{dq_2}r_1(q_2)\big|_{q_2=0}\,,
\ee
which, bearing in mind (\ref{eq:derivate-r}), amounts to ask (\ref{eq:sfercritica_h=0-II}).

Since $r_2$ is increasing around the origin, we have $\partial^2_{q_1}\RS<0$ and by the implicit function theorem we define locally a function $q_1(q_2)$ increasing and positive, vanishing at the origin. We set
\be\label{eq:RS1-gauss}
\RS_1(q_2):=\max_{q_1}\RS(q_1,q_2)=\RS(q_1(q_2),q_2)\,
\ee
and compute
\be\label{eq:derRS1gauss}
\partial_{q_2}\RS_1(q_2)=\frac{\b^2\a(1-\a)}{2}\left(q_1(q_2)-r_1(q_2)\right)\,. 
\ee
By (\ref{eq:sfercritica_h=0}) we have $\partial^2_{q_2}\RS_1\big|_{q_2=0}>0\,$, whence we obtain (\ref{eq:origin1}). 

If (\ref{eq:sfercritica_h=0-III}) holds, then 
\be\label{eq:sfercritica_h=0-IIIbis}
\frac{d}{dq_2}q_1(q_2)\big|_{q_2=0}<\frac{d}{dq_2}r_1(q_2)\big|_{q_2=0}\,,
\ee
which proceeding as before leads to (\ref{eq:origin2}). 

However also in the case $b_1=b_2=0$ we can repeat all the steps done in the case $b_1^2+b_2^2>0$, showing the existence of a point $(\bar q_1, \bar q_2)$ in which a $\min\max$ of $\RS$ is attained. If \eqref{eq:sfercritica_h=0} (\ie (\ref{eq:sfercritica_h=0-II})) holds then it must be $(\bar q_1, \bar q_2)=(0,0)$. If \eqref{eq:sfercritica_h=0-III} holds, then 
\eqref{eq:sfercritica_h=0-IIIbis} enforces 
$$
q_1(q_2)-r_1(q_2)\leq0
$$
in a neighbourhood of the origin (as $q_1(0)=r_1(0)=0$), which implies that the critical point $(\bar q_1, \bar q_2)$ must fall elsewhere. 
\end{proof}

\begin{proof}[Proof of Lemma \ref{lemma_Gamma-as}]
We will prove that for all $u\in \sqrt qS^N$
\be\label{eq:1bodysf}
\G^{(\s)} (q):=\lim_N \G_N (u)=\frac12\min_{r\in[0,1)}\left(q(1-r)+\frac {r}{1-r}+\log(1-r)\right)\,.
\ee
We show first that (\ref{eq:1bodysf}) implies the assertion.
Let $h$ be a random vector with i.i.d.\ $\mathcal N(0,q)$ entries. (As customary we write $X\simeq Y$ if there are constants $c,C>0$ such that $cY\leq X\leq CY$). The classical estimates
\be\label{eq:normConc}
\Gamma_N(h)\leq \frac{\|h\|_2}{\sqrt N}\,,\quad P\left(\left|\frac{\|h\|_2}{\sqrt N}-\sqrt{q}\right|\geq t\right)\simeq e^{-\frac{t^2N}{2}}\,
\ee
permit us to write for all $t>0$ (small)
\bea
|E[\G_N]-\G^{(\s)} (q)|&\leq& |E[\G_N1_{\left\{\left|\frac{\|h\|}{\sqrt N}-\sqrt{q}\right|< t\right\}}]-\G^{(\s)} (q)|+\left|E\left[\frac{\|h\|_2}{\sqrt N}1_{\left\{\left|\frac{\|h\|}{\sqrt N}-\sqrt{q}\right|\geq t\right\}}\right]\right|\nn\\
&\simeq&\left|\G_N(u^*)P\left(\left|\frac{\|h\|_2}{\sqrt N}-\sqrt{q}\right|< t\right)-\G^{(\s)} (q)\right|+o(t)+e^{-t^2N/2}\,\nn\\
&\simeq&\left|\G_N(u^*)-\G^{(\s)} (q)\right|+o(t)+e^{-t^2N/2}\,,
\eea
for some $u^*\in \sqrt qS^{N}$ and $o(t)\to0$ as $t\to0$. Since $t>0$ is arbitrary we obtain
$$
|E[\G_N]-\G^{(\s)} (q)|\leq\left|\G_N(u^*)-\G^{(\s)} (q)\right|\,. 
$$
It remains to show (\ref{eq:1bodysf}). Given $\e>0$ we introduce the spherical shell
$$
S^{N,\e}:=S^N+\sqrt{\frac{ \e}{N}} S^N
$$
and the measure $\s_N^{(\e)}$ as the uniform probability on it. For any $\theta>0$ we have
\bea
\int \s^{(\e)}_{N}(dx) e^{(u,x)}&\leq& e^{\frac{\theta (N+\e)}{2}}\int \s^{(\e)}_{N}(dx) e^{-\frac{\theta}{2}\|x\|_2^2+(u,x)}\nn\\
&\leq&e^{\frac{\theta (N+\e)}{2}}\frac{\sqrt{2\pi}^N}{\theta^{\frac N2}|S^{N,\e}|}\int e^{-\frac{\theta}{2}\|x\|_2^2+(u,x)}\frac{dx}{\sqrt{2\pi}^N}\nn\\
&=&e^{\frac{\theta (N+\e)}{2}+\frac{q N}{2\theta}}\frac{\sqrt{2\pi}^N}{\theta^{\frac N2}|S^{N,\e}|}\,.
\eea
Therefore for $C>0$ large enough
\be
\frac1N\log \int \s^{(\e)}_{N}(dx) e^{(u,x)}\leq \frac {\theta }{2}+\frac{q}{2\theta}-\frac12(\log\theta+1)+C\theta\frac\e N\,.
\ee
Since this inequality holds for all $\theta>0$ and $\e>0$ we have
\be\label{eq:Gamma-lim-sup}
\limsup_N \G_N (u)\leq \inf_{\theta>0}\left(\frac{q}{2\theta}+\frac {\theta -1}{2}-\frac12\log\theta\right)\,.
\ee
We set for brevity 
$$
\Gamma_1(\theta):=\frac{q}{2\theta}+\frac {\theta -1}{2}-\frac12\log\theta\,
$$
and notice that $\Gamma_1$ is uniformly convex in all the intervals $(0,\theta_0)$ for finite $\theta_0>0$. 

For the reverse bound, again we let $\theta>0$ and write
\be\label{eq:rev-bound-dec}
\int \s^{(\e)}_{N}(dx) e^{(u,x)}\geq e^{\frac\theta2 N}\int_{\R^N} \frac{dx}{|S^{N,\e}|} e^{-\frac\theta2\|x\|_2^2+(u,x)}-e^{\frac\theta2 N}\int_{(S^{N,\e})^c} \frac{dx}{|S^{N,\e}|} e^{-\frac\theta2\|x\|_2^2+(u,x)}\,.
\ee
The first summand on the r.h.s.\ can be written as before
\be\label{eq:rev-bound-1}
e^{\frac\theta2 N}\int_{\R^N} \frac{dx}{|S^{N,\e}|} e^{-\frac\theta2\|x\|_2^2+(u,x)}=e^{\frac{\theta N}{2}+\frac{q N}{2\theta}}\frac{\sqrt{2\pi}^N}{\theta^{\frac N2}|S^{N,\e}|}\,.
\ee
For the second summand we introduce $\eta\in(0,\frac\theta2)$ and bound
\bea
e^{\frac\theta2 N}\int_{\|x\|^2\leq N-\e} \frac{dx}{|S^{N,\e}|} e^{-\frac\theta2\|x\|_2^2+(u,x)}&\leq&e^{\frac\theta2 N+(N-\e)\frac\eta2+\frac{q N}{2(\theta+\eta)}} \frac{\sqrt{2\pi}^N}{\theta^{\frac N2}|S^{N,\e}|}\label{eq:rev.IIsum1} \\
e^{\frac\theta2 N}\int_{\|x\|^2\geq N+\e} \frac{dx}{|S^{N,\e}|} e^{-\frac\theta2\|x\|_2^2+(u,x)}&\leq&e^{\frac\theta2 N-(N+\e)\frac\eta2+\frac{q N}{2(\theta-\eta)}}\frac{\sqrt{2\pi}^N}{\theta^{\frac N2}|S^{N,\e}|} \label{eq:rev.IIsum2}\,.
\eea
Thus
\be\label{eq:USAREV}
\liminf_N\frac1N\log\int \s^{(\e)}_{N}(dx) e^{(u,x)}\geq \max(\Gamma_1, \Gamma_2,\Gamma_3) 
\ee
with
\bea
\Gamma_2(\eta,\theta)&:=&\frac{q}{2(\theta-\eta)}+\frac {\eta (1-\frac\e N)}{2}+\frac {\theta-1}{2}-\frac12\log\theta\,,\nn\\
\Gamma_3(\eta,\theta)&:=&\frac{q}{2(\theta+\eta)}-\frac {\eta(1+\frac\e N)}{2}+\frac {\theta-1}{2}-\frac12\log\theta\nn\,.
\eea
Now we define
\be
\D_{12}(\eta,\theta):=\Gamma_1(\theta)-\Gamma_2(\eta,\theta)\,,\quad  \D_{13}(\eta,\theta):=\Gamma_1(\theta)-\Gamma_3(\eta,\theta)\,,
\ee
and we seek $\bar\theta>0$ for which $\D_{12},\D_{13}\geq0$ for sufficiently small $\eta$. Since $\D_{12}(0,\theta)=\D_{13}(0,\theta)=0$ it suffices to study 
\be\label{eq:zero-REV}
\frac{d}{d\eta} \D_{12}\Big|_{\eta=0}\,,\quad\frac{d}{d\eta} \D_{13}\Big|_{\eta=0}\,.
\ee
A direct computation shows
\bea
\frac{d}{d\eta} \D_{12}\Big|_{\eta=0}&=&\frac{\e}{2N}-\partial_\theta \G_1(\theta)\label{eq:uno-REV}\,,\\
\frac{d}{d\eta} \D_{13}\Big|_{\eta=0}&=&\frac{\e}{2N}+\partial_\theta \G_1(\theta)\label{eq:due-REV}\,.
\eea
Combining (\ref{eq:zero-REV}), (\ref{eq:uno-REV}) and (\ref{eq:due-REV}) we see that plugging $\bar\theta=\arg\min \G_1$ into \eqref{eq:USAREV} we arrive to
\be\label{eq:lim-inf}
\lim\inf_N \G_N (u)\geq \min_{\theta>0}\left(\frac{q}{2\theta}+\frac {\theta -1}{2}-\frac12\log\theta\right)\,.
\ee
Therefore (\ref{eq:Gamma-lim-sup}) and (\ref{eq:lim-inf}) give
$$
\lim_N \G_N (u)= \min_{\theta>0}\left(\frac{q}{2\theta}+\frac {\theta -1}{2}-\frac12\log\theta\right)\nn\\
$$
and changing variable $\theta=(1-r)^{-1}$ we obtain (\ref{eq:1bodysf}). 
\end{proof}

{\bf Acknowledgements:} This manuscript benefited greatly from the observations of two anonymous referees, who are gratefully acknowledged.

\



\begin{thebibliography}{00}
%
\bibitem{AC} A. Auffinger, W.-K. Chen, {\em Free energy and complexity of spherical bipartite models}, J. Stat. Phys. 157, 40-59, (2014).
%
\bibitem{baik} J. Baik, J. O. Lee {\em Free energy of bipartite spherical Sherrington-Kirkpatrick model}, Ann. Inst. H. Poincar\'e Probab. Statist. 56(4): 2897-2934 (2020).
%
\bibitem{NN} A. Barra, G. Genovese, F. Guerra, {\em The Replica Symmetric Behaviour of the Analogical Neural Network}, J. Stat. Phys. 142, 654, (2010).
%
\bibitem{bip} A. Barra, G. Genovese, F. Guerra, \textit{Equilibrium statistical mechanics of bipartite spin systems}, J. Phys. A: Math. Theor. \textbf{44}, 245002 (2011).
%
\bibitem{bates1} E. Bates, Y. Sohn {\em Free energy in multi-species mixed $p$-spin spherical models} 	arXiv:2109.14790 (2021).
%
\bibitem{bates2} E. Bates, Y. Sohn {\em Crisanti-Sommers formula and simultaneous symmetry breaking in multi-species spherical spin glasses}, arXiv:2109.14791 (2021).
%
\bibitem{GG} G. Genovese, \textit{Minimax formula for the replica symmetric free energy of deep restricted Boltzmann machines}, (2020).
%
\bibitem{sferico0} G. Genovese, D. Tantari {\em Legendre Duality of Spherical and Gaussian Spin Glasses}, Math. Phys. Anal. Geom. 18, 1, (2015).
\bibitem{sferico} G. Genovese, D. Tantari {\em Legendre Equivalences of Spherical Boltzmann Machines}, in Journal Physics A, special issue Machine learning and statistical physics, theory, inspiration, application, Ed. E. Agliari, A. Barra, P. Sollich, L. Zdeborova, (2020). 
%
\bibitem{tucaphd} P. Kivimae, {\em The Ground State Energy and Concentration of Complexity in Spherical Bipartite Models},  arXiv:2107.13138 (2021).
%
\bibitem{mc} B. McKenna, {\em Complexity of bipartite spherical spin glasses}, arXiv:2105.05043 (2021).
%
\bibitem{pan-sf} D. Panchenko, {\em Cavity method in the spherical SK model}, Ann. Inst. H. Poincar\'e Probab. Statist. 45(4): 1020-1047 (2009).
%
\bibitem{Tal} M. Talagrand, \emph{Mean Field Models for Spin Glasses}, Vol. 1, Springer-Verlag Berlin Heidelberg (2011).
%
\bibitem{DT} D. Tantari, private communication. 
%
\end{thebibliography}
\end{document}